\documentclass[reqno]{amsart}
\usepackage{amscd,amsfonts,amssymb}
\usepackage{graphicx}
\usepackage{color}

\copyrightinfo{2000}{American Mathematical Society}

\newcommand{\cD}{{\mathcal D}}

\newcommand{\cK}{{\mathcal K}}

\newcommand{\bR}{{\mathbb R}}

\numberwithin{equation}{section}

\newtheorem{theorem}{Theorem}[section]
\newtheorem{lemma}{Lemma}[section]
\newtheorem{corollary}{Corollary}[section]
\theoremstyle{definition}
\newtheorem{definition}{Definition}[section]
\theoremstyle{remark}
\newtheorem{remark}{Remark}[section]

\author{S.~Albeverio}
\address{Institut f\"{u}r Angewandte Mathematik, Universit\"{a}t Bonn,
Endenicher Allee 60, D-53115 Bonn, Germany; \
SFB 611, Bonn University; IZKS, Bonn University; BiBoS (Bie\-lefeld-Bonn);
Dip. Matematika, Universita di Trento}
\email{albeverio@uni-bonn.de}

\author{O.~S.~Rozanova}
\address{Mathematics and Mechanics Faculty, Moscow State University, Moscow
119992, Russia.}
\email{rozanova@mech.math.msu.su}

\author{V.~M.~Shelkovich}
\address{Department of Mathematics, St.-Petersburg State Architecture and
Civil Engineering University,
2 Krasnoarmeiskaya 4, 190005, St. Petersburg, Russia.}
\email{shelkv@vs1567.spb.edu}

\title[Transport processes in the zero-pressure gas dynamics]
{Transport and concentration processes in the multidimensional zero-pressure gas
dynamics model with the energy conservation law}

\thanks{The authors were supported by DFG Project 436 RUS 113/895. The second and third authors
(O.R. and V.S.) were also supported by the Analytical departmental special program
"The development of scientific potential of the Higher School", project 2.1.1/1399.}



\date{ }

\begin{document}

\begin{abstract}
We introduce integral identities to define $\delta$-shock wave
type solutions for the multidimensional zero-pressure gas dynamics
$$
\begin{array}{rcl}
\displaystyle
\rho_t + \nabla\cdot(\rho U)&=&0, \\
\displaystyle
(\rho U)_t + \nabla\cdot(\rho U\otimes U)&=&0, \\
\displaystyle
\Big(\frac{\rho|U|^2}{2}+H\Big)_t
+\nabla\cdot\Big(\Big(\frac{\rho|U|^2}{2}+H\Big)U\Big)&=&0, \\
\end{array}
$$
where where $\rho$ is the density, $U\in\bR^n$
is the velocity, $H(x,t)$ is the internal energy, $x\in\bR^n$.
Using these integral identities, the Rankine-Hugoniot conditions for $\delta$-shocks
are obtained.
We derive the balance laws describing mass, momentum, and energy transport
from the area outside the $\delta$-shock wave front onto this front.
These processes are going on in such a way that the total mass, momentum, and energy are conserved
and at the same time mass and energy of the moving $\delta$-shock wave front are increasing quantities.
In addition, the total kinetic energy transfers into the total internal energy.
The process of propagation of $\delta$-shock waves is also described.
These results can be used in modeling of mediums which
can be treated as a {\em pressureless continuum} (dusty gases,
two-phase flows with solid particles or droplets, granular gases).
\end{abstract}

\maketitle

\section{Strong singular solutions and pressureless mediums}
\label{s1}

\subsection{$L^\infty$-type solutions.}\label{s1.1}
Let us recall some classical results. Consider the Cauchy problem
for the system of conservation laws in one dimension space:
\begin{equation}
\label{1.1}
\left\{
\begin{array}{rclrcl}
U_t+\big(F(U)\big)_x&=&0, &&\text{in}\quad \bR\times (0, \ \infty),
\medskip \\
U&=&U^0,                  &&\text{in} \quad \bR\times \{t=0\}, \\
\end{array}
\right.
\end{equation}
where $F:\bR^m \to \bR^m$ is called the {\em flux-function}
associated with (\ref{1.1}); $U^0:\bR \to \bR^m$ are given
vector-functions; $U=U(x,t)=(u_1(x,t),\dots,u_m(x,t))$
is the unknown function with value in $\bR^m$, and components $u_j(x,t)$,
$j=1,\dots,m$; $x\in\bR$, $t\ge 0$.

As is well known, even in the case of smooth (and, certainly, in the case of discontinuous)
initial data $U^0(x)$, in general, does not exist any smooth and global in time solution
of system (\ref{1.1}). As noted in the Evans' book~\cite[11.1.1.]{Evans},
``the great difficulty in this subject is discovering a proper notion of weak solution
for the initial problem (\ref{1.1})''. ``We must devise some way to interpret
a less regular function as somehow ``solving'' this initial-value
problem''~\cite[3.4.1.a.]{Evans}.
But it is a well known that a partial differential equation
may not make sense even if $U$ is differentiable.
``However, observe that if we {\em temporarily} assume $U$ is smooth, we can as follows
rewrite, so that the resulting expression does not directly involve the derivatives
of $U$''~\cite[3.4.1.a.]{Evans}. ``The idea is to multiply the partial differential equation
in (\ref{1.1}) by a smooth function $\varphi$ and then to integrate by parts, thereby
transferring the derivatives onto $\varphi$''~\cite[3.4.1.a.;11.1.1.]{Evans}.
Following this suggestion we shall derive the integral identity which gives the following definition
of an $L^\infty$-{\em generalized solution} of the Cauchy problem (\ref{1.1}):
$U\in L^\infty\big(\bR\times (0,\infty);\bR^m\big)$ is called a
{\em generalized solution} of the Cauchy problem (\ref{1.1})
if the integral identity
\begin{equation}
\label{2.1}
\int_{0}^{\infty}\int \Big(U\cdot{\widetilde\varphi}_{t}
+F(U)\cdot{\widetilde\varphi}_{x}\Big)\,dx\,dt
+\int U^0(x)\cdot{\widetilde\varphi}(x,0)\,dx=0
\end{equation}
holds for all compactly supported smooth test vector-functions
${\widetilde\varphi}: \bR\times [0, \infty) \to \bR^m$, where \
$\cdot$ \ is the scalar product of vectors, and $\int f(x)\,dx$
denotes the improper integral $\int_{-\infty}^{\infty}f(x)\,dx$.
``This identity, which we derived supposing $U$ to be a smooth
solution makes sense if $U$ is merely
bounded''~\cite[11.1.1.]{Evans}.

\begin{theorem}
\label{th0*} {\rm (see, e.g.,~\cite[11.1.1.]{Evans})} Let
$\Omega\subset \bR\times (0, \infty)$ be a region cut by a smooth
curve $\Gamma$ into a left- and right-hand parts $\Omega_{\mp}$. Let
us assume that the generalized solution $U$ of {\rm (\ref{1.1})} is
smooth on either side of the curve $\Gamma$ along which $U$ has
simple jump discontinuities. Then the Rankine--Hugoniot condition
\begin{equation}
\label{2.1*}
\big[F(U)\big]_{\Gamma}\nu_1+\big[U\big]_{\Gamma}\nu_2=0,
\end{equation}
holds along $\Gamma$, where ${\bf n}=(\nu_1,\nu_2)$ is the unit normal
to the curve $\Gamma$ pointing from $\Omega_{-}$ into $\Omega_{+}$,
$$
[F(U)]\stackrel{def}{=}F(U_{-})-F(U_{+}),
$$
$[U]\stackrel{def}{=}U_{-}-U_{+}$ are the jumps in $F(U)$ and in $U$ across
the discontinuity curve $\Gamma$, respectively.
$U_{\mp}$ are respective the left- and right-hand values of $U$ on
$\Gamma$.

If $\Gamma=\{(x,t): x=\phi(t)\}$, where $\phi(\cdot)\in C^1(0,+\infty)$,
then
\begin{equation}
\label{31}
{\bf n}=(\nu_1,\nu_2)
=\frac{1}{\sqrt{1+(\dot\phi_i(t))^2}}\big(1,-\dot\phi_i(t)\big),
\end{equation}
and {\rm (\ref{2.1*})} reads
\begin{equation}
\label{2.1**}
\big[F(U)\big]_{\Gamma}=\dot\phi(t)\big[U\big]_{\Gamma},
\end{equation}
where $\dot{(\cdot)}=\frac{d}{dt}(\cdot)$.
\end{theorem}

It is well known that if $U\in L^\infty\big(\bR\times (0,\infty);\bR^m\big)$
is a generalized solution of the Cauchy problem (\ref{1.1}) compactly
supported with respect to~$x$, then the integral of the solution on
the whole space
\begin{equation}
\label{7.0}
\int U(x,t)\,dx=\int U^0(x)\,dx, \qquad t\ge 0
\end{equation}
is independent of time. These integrals can express the conservation
laws of quantities like the {\em total area, mass, momentum, energy}, etc.

\subsection{$\delta$-shocks.}\label{s1.2}
It is well known that there are ``nonclassical'' situations where, in
contrast to Lax's and Glimm's classical results, the Cauchy problem for
a system of conservation laws {\em either does not possess a weak $L^{\infty}$-solution
or possesses it for some particular initial data}. In order to solve the Cauchy
problem in these ``nonclassical'' situations, it is necessary to seek solutions of
this Cauchy problem in class of singular solutions called {\em $\delta$-shocks}.
Roughly speaking, a {\em $\delta$-shock} is a solution such that its
components {\em contain Dirac delta-functions}.

It is customary to assume that a $\delta$-shock wave type solution was first
described by Korchinski in his unpublished dissertation~\cite{Korchinski} in 1977.
However, in fact, a solution of this type as well as the Rankine-Hugoniot condition
for the one-dimensional continuity equation were already derived from
physical considerations in the book~\cite[\S7,\S12]{Zeldovich-Mishkis} in 1973.
Next, in 1979, A.~N.~Kraiko~\cite{Kraiko} considered {\em a new type
of discontinuity surface which are to be introduced in certain models of media having
no inherent pressure} and obtained the Rankine-Hugoniot conditions for them. The
system under consideration in~\cite{Kraiko} is the zero-pressure gas dynamics
described by the system of equations:
\begin{equation}
\label{11.1-0***}
\begin{array}{rcl}
\displaystyle
\rho_t+\big(\rho u\big)_x&=&0, \\
\displaystyle
(\rho u)_t+\big(\rho u^2\big)_x&=&0, \\
\displaystyle
\Big(\frac{\rho u^2}{2}+\rho\tau\Big)_t+\Big(\Big(\frac{\rho u^2}{2}+\rho\tau\Big)u\Big)_x&=&0, \\
\end{array}
\end{equation}
where $\rho(x,t)\ge 0$ is the density, $u(x,t)$ is the velocity, $\rho(x,t)u(x,t)$ is corresponding
momentum, $\tau(x,t)$ is the {\em internal energy per unit mass}, $x\in\bR$.
The last system can be derived from the Euler equations of nonisentropic gas dynamics
\begin{equation}
\label{11.1-2-p}
\rho_t+\big(\rho u\big)_x=0, \quad (\rho u)_t+\big(\rho u^2+p\big)_x=0,
\quad (\rho E)_t+\big((\rho E+p)u\big)_x=0,
\end{equation}
if we set $p=0$, where $E=\frac{\rho u^2}{2}+\tau$ is {\em total
energy per unit mass}.

According to~\cite[page~502]{Kraiko}, to construct a solution for
system (\ref{11.1-0***}) for arbitrary initial data, we need
discontinuities which would be different from classical ones and
{\em carry mass, impulse and energy}. As it turned out these
nonclassical discontinuities are {\em $\delta$-shocks}.

The theory of $\delta$-shocks has been intensively developed
in the last fifteen years (for example, see~\cite{Al-S},~\cite{B}--~\cite{E-R-S},
~\cite{Li-Zh1}--~\cite{Ned-Oberg},~\cite{umn}--~\cite{Y-1} and the references therein).
Moreover, recently, in~\cite{Pan-S1}, a
concept of $\delta^{(n)}$-shock wave type solutions was introduced,
$n=1,2,\dots$. It is a {\it new type of singular solution} of a system of
conservation laws such that its components contain delta functions and
their derivatives up to $n$-th order. In~\cite{Pan-S1},~\cite{S5}, the
theory of {\em $\delta'$-shocks} was established. The results~\cite{Pan-S1} and
~\cite{S5} show that systems of conservation laws can
develop solutions not only of the type of Dirac measures (as in the case of $\delta$-shocks)
but also the type of derivatives of such measures.

The above-mentioned singular solutions do not satisfy the standard integral identities
of the type (\ref{2.1}). To define them we use special integral identities and derive
special Rankine--Hugoniot conditions.
These solutions are connected with {\em transport and concentration processes}
~\cite{Al-S},~\cite{Chen-Liu2},~\cite{Pan-S1},~\cite{umn},~\cite{S8}.

In the numerous papers cited above $\delta$-shocks were studied for the
{\em system of zero-pressure gas dynamics}:
\begin{equation}
\label{g-2}
\rho_t + \nabla\cdot(\rho U)=0, \qquad
(\rho U)_t + \nabla\cdot(\rho U\otimes U)=0,
\end{equation}
where $\rho=\rho(x,t)\ge 0$ is the density, $U=(u_1(x,t),\dots,u_n(x,t))\in\bR^n$
is the velocity,
$\nabla=\big(\frac{\partial}{\partial x_1},\dots,\frac{\partial}{\partial x_n}\big)$,
\ $\cdot$ \ is the scalar product of vectors, $\otimes$ is the usual tensor
product of vectors.

The system of zero-pressure gas dynamics (\ref{g-2}) has a physical context and is used in applications.
This system can be considered as a model of the ``sticky particle dynamics'' and was used, e.g.,
to describe the formation of large-scale structures of the universe~\cite{San-Z},~\cite{Z},
for modeling the formation and evolution of traffic jams~\cite{B-D-D-R},
for modeling non-classical shallow water flows~\cite{Oc-Oc}.
Nonlinear equations (in particular, zero-pressure gas dynamics) admitting
$\delta$-shock wave type solutions are appropriate
for modeling and studying singular problems like movement of multiphase
media (dusty gases, two-phase flows with solid particles or droplets).
The presence of particles or droplets may drastically modify flow
parameters. Moreover, a large number of phenomena that are absent in
pure gas flow is inherent in two-phase flows. Among them there are
local accumulation and focusing of particles, inter-particle and
particle-wall collisions resulting in particle mixing and dispersion,
surface erosion due to particle impacts, and particle-turbulence
interactions which govern the dispersion and concentration heterogeneities
of inertial particles. The dispersed phase is usually treated mathematically
as a {\em pressureless continuum}. Models of such media were discussed
in the papers~\cite{Kraiko}--~\cite{Kraiko-2},~\cite{Osiptsov-1}
--\cite{Osiptsov-3}.
Equations admitting $\delta$-shocks can also be used for modeling granular gases.
Granular gases are dilute assemblies of hard spheres which lose energy at collisions.
In such gases local density excesses and local pressure falls~\cite{F-Meerson-2},~\cite{F-Meerson-3}.
In~\cite{F-Meerson-2},~\cite{F-Meerson-3}, the following hydrodynamics system
of granular gas
$$
\rho_t+(\rho u)_x=0, \quad \rho(u_t+uu_x)=-(\rho T)_x,
\quad T_t+uT_x=-(\gamma-1)Tu_x-\Lambda\rho T^{3/2},
$$
was studied, where $\rho$ is the gas density, $u$ is the velocity, $T$ is the temperature,
$\gamma$ is the adiabatic index, $p=\rho T$ is the pressure. It was shown that
for {\em non-zero pressure} this system admits a solution which contains a $\delta$-function in
the density $\rho$.

\subsection{Main results.}\label{s1.3}
As it follows from~\cite{Kraiko}--~\cite{Kraiko-2}, for modeling media which
can be considered as {\em having no pressure} we must take into account energy
transport. In the above-cited papers zero-pressure gas dynamics was studied only
in the form (\ref{g-2}). Therefore, we need to study $\delta$-shocks in zero-pressure gas dynamics
\begin{equation}
\label{g-2-en}
\begin{array}{rcl}
\displaystyle
\rho_t + \nabla\cdot(\rho U)&=&0, \\
\displaystyle
(\rho U)_t + \nabla\cdot(\rho U\otimes U)&=&0, \\
\displaystyle
\Big(\frac{\rho|U|^2}{2}+H\Big)_t
+\nabla\cdot\Big(\Big(\frac{\rho|U|^2}{2}+H\Big)U\Big)&=&0, \\
\end{array}
\end{equation}
where $H(x,t)$ is the internal energy, $|U|^2=\sum_{k=1}^nu_k^2$. This system
is obtained by adding an energy conservation law to zero-pressure gas dynamics
(\ref{g-2}). As distinct from (\ref{11.1-0***}) it is more convenient
for us to consider as a variable $H$ instead of $H=\rho \tau$, where
$\tau$ is the internal energy per unit mass.
The reason is that since for singular solution $\tau(x,t)$ and $\rho(x,t)$
must contain $\delta$-functions, it is impossible to define the product $\tau(x,t)\rho(x,t)$.

Under the second thermodynamics law it is natural to supplement the
system (\ref{11.1-0***}) with a state equation $\tau=\tau(T)$, where
$T$ is the temperature.  For \eqref{g-2-en} the natural state
equation is $H=H(\rho,T),$ moreover, $H(0,T)=0$.

In Sec.~\ref{s2}, we introduce Definition~\ref{g-de3-1-en} of $\delta$-shock
wave type solutions for system (\ref{g-2-en}). Next, using this definition,
by Theorem~\ref{th1-10-en} we derive the corresponding Rankine-Hugoniot
conditions for $\delta$-shocks (\ref{g-51-10*-en}). These Rankine-Hugoniot
conditions are the direct analog of those that were introduced by A.~N.~Kraiko~\cite{Kraiko}.

In Sec.~\ref{s3}, we show that $\delta$-shocks are related with the
{\em transport processes} of mass, momentum and energy.
According to Theorems~\ref{g-th5-en},~\ref{e-th4-en}, the mass, momentum
and energy {\em transport processes} between the area outside of the moving
$\delta$-shock wave front and this front {\em are going on}
such that the total mass, momentum and energy are independent of time.
Moreover, the {\em mass and energy concentration processes} takes place on
the $\delta$-shock wave front.

\section{$\delta$-shock type solutions and the Rankine--Hugoniot conditions}
\label{s2}

\subsection{$\delta$-shock type solutions.}\label{s2.1}
Throughout the paper we shall systematically use some results recalled in Appendix~\ref{s6}.
Let $\Gamma=\bigl\{(x,t): S(x,t)=0\bigr\}$ be a hypersurface of codimension~1
in the upper half-space $\{(x,t): x\in \bR^n, \ t\in [0,\infty)\}\subset\bR^{n+1}$,
$S\in C^{\infty}(\bR^{n}\times[0,\infty))$, with $\nabla S(x,t)\bigr|_{S=0}\neq 0$
for any fixed $t$, where
$\nabla=\big(\frac{\partial}{\partial x_1},\dots, \frac{\partial}{\partial x_n}\big)$.
Let $\Gamma_t=\bigl\{x\in \bR^n: S(x,t)=0\bigr\}$ be a moving surface in $\bR^n$.
Denote by $\nu$ the unit space normal to the surface $\Gamma_t$ pointing
(in the positive direction) from $\Omega^{-}_{t}=\{x\in \bR^{n}: S(x,t)<0\}$
to $\Omega^{+}_{t}=\{x\in \bR^{n}: S(x,t)>0\}$ such that
$\nu_j=\frac{S_{x_j}}{|\nabla S|}$, \ $j=1,\dots,n$.
The direction of the vector $\nu$ coincides with the direction in which the
function $S$ increases, i.e., inward the domain $\Omega^{+}_t$.
The time component of the normal vector $-G=\frac{S_t}{|\nabla S|}$ is the
{\em velocity of the wave front} $\Gamma_t$ along the space normal $\nu$.

For system (\ref{g-2-en}) we consider the {\em $\delta$-shock type initial data}
\begin{equation}
\label{5.0-10-en}
\begin{array}{rcl}
\displaystyle
\bigl(U^0(x), \rho^0(x), H^0(x), \,x\in \bR^{n}; \, U^{0}_{\delta}(x), \, x\in \Gamma_0\bigr), &&\\
\displaystyle
\text{where} \quad \rho^0(x)&=&{\widehat \rho}^0(x)+e^0(x)\delta(\Gamma_0), \\
\displaystyle
H^0(x)&=&{\widehat H}^0(x)+h^0(x)\delta(\Gamma_0),
\end{array}
\end{equation}
such that $U^0\in L^\infty\big(\bR^n;\bR^n\big)$,
${\widehat \rho}^0,{\widehat H}^0 \in L^\infty\big(\bR^n;\bR\big)$,
$e^0,h^0\in C(\Gamma_0)$, \ $\Gamma_0=\bigl\{x: S^0(x)=0\bigr\}$ is the
initial position of the $\delta$-shock front,
$\nabla S^0(x)\bigr|_{S^0=0}\neq 0$, \ $U^{0}_{\delta}(x)$, $x\in \Gamma_0$,
is the {\em initial velocity} of the $\delta$-shock, $\delta(\Gamma_0)$
($\equiv\delta(S^0)$) is the Dirac delta function concentrated on the
surface $\Gamma_0$ defined by (\ref{g-106}):
$$
\big\langle \delta(S_0), \ \varphi(x) \big\rangle=\int_{\Gamma_0}\varphi(x)\,d\Gamma_0,
\quad \forall \, \varphi\in {\cD}(\bR^n),
$$
$d\Gamma_0$ is the surface measure on the surface $\Gamma_0$.

Similarly to~\cite[Definition~9.1.]{umn} we introduce the following definition

\begin{definition}
\label{g-de3-1-en} \rm
A triple of distributions $(U,\rho, H)$ and a hypersurface $\Gamma$,
where $\rho(x,t)$ and $H(x,t)$ have the form of the sum
$$
\rho(x,t)={\widehat \rho}(x,t)+e(x,t)\delta(\Gamma),
\quad
H(x,t)={\widehat H}(x,t)+h(x,t)\delta(\Gamma),
$$
and
$U\in L^\infty\big(\bR^n\times(0,\infty);\bR^n\big)$,
${\widehat \rho},{\widehat H} \in L^\infty\big(\bR^n\times(0,\infty);\bR\big)$,
$e,h\in C(\Gamma)$, is called a {\em $\delta$-shock wave type solution}
of the Cauchy problem (\ref{g-2-en}), (\ref{5.0-10-en}) if the integral identities
\begin{equation}
\label{g-4.0-10-en}
\begin{array}{rcl}
\displaystyle
\int_{0}^{\infty}\int
{\widehat \rho}\Big(\varphi_{t} + U\cdot\nabla\varphi\Big)\,dx\,dt
+\int_{\Gamma}e \frac{\delta\varphi}{\delta t}\frac{\,d\Gamma}{\sqrt{1+G^2}}
\qquad\qquad && \medskip \\
\displaystyle
+\int {\widehat\rho}^0(x)\varphi(x,0)\,dx
+\int_{\Gamma_0}e^0(x)\varphi(x,0)\,d\Gamma_{0}&=&0, \medskip \\
\displaystyle
\int_{0}^{\infty}\int
{\widehat \rho}U\Big(\varphi_{t} + U\cdot\nabla\varphi\Big)\,dx\,dt
+\int_{\Gamma}e U_{\delta}\frac{\delta\varphi}{\delta t}\frac{\,d\Gamma}{\sqrt{1+G^2}}
\qquad\quad && \medskip \\
\displaystyle
+\int U^0(x){\widehat\rho}^0(x)\varphi(x,0)\,dx
+\int_{\Gamma_0}e^0(x)U_{\delta}^0(x)\varphi(x,0)\,d\Gamma_{0}&=&0, \medskip \\
\displaystyle
\int_{0}^{\infty}\int \bigg(\Big(\frac{{\widehat \rho}|U|^2}{2}+{\widehat H}\Big)\varphi_{t}
+\Big(\frac{{\widehat \rho}|U|^2}{2}+{\widehat H}\Big)U\cdot\nabla\varphi\bigg)\,dx\,dt
\qquad && \medskip \\
\displaystyle
+\int_{\Gamma}\Big(\frac{e|U_{\delta}|^2}{2}+h\Big)
\frac{\delta\varphi}{\delta t}\frac{\,d\Gamma}{\sqrt{1+G^2}}
\qquad\qquad\qquad\qquad && \medskip \\
\displaystyle
+\int\Big(\frac{{\widehat \rho^0(x)}|U^0(x)|^2}{2}+{\widehat H}^0(x)\Big)\varphi(x,0)\,dx
\qquad\quad && \medskip \\
\displaystyle
\qquad
+\int_{\Gamma_0}\Big(\frac{e^0(x)|U_{\delta}^0(x)|^2}{2}+h^0(x)\Big)\varphi(x,0)\,d\Gamma_{0}&=&0, \\
\end{array}
\end{equation}
hold for all $\varphi \in {\cD}(\bR^n\times [0, \infty))$.
Here $\int f(x)\,dx$ denotes the improper integral $\int_{\bR^n}f(x)\,dx$;
$d\Gamma$ and $d\Gamma_0$ are the surface measures on the surfaces $\Gamma$ and $\Gamma_0$, respectively;
\begin{equation}
\label{g-4.0-10*}
U_{\delta}=\nu G=-\frac{S_t\nabla S}{|\nabla S|^2}
\end{equation}
is the $\delta$-shock velocity, $\nu$ is the unit space normal to the surface $\Gamma_t$
introduced above; $-G=\frac{S_t}{|\nabla S|}$,
$\frac{\delta\varphi}{\delta t}$ is the $\delta$-derivative with
respect to the time variable (\ref{g-74});
$\delta(\Gamma)$ is the Dirac delta function concentrated on the
surface $\Gamma$ defined by (\ref{g-106}):
$$
\big\langle \delta(S), \ \varphi(x,t) \big\rangle
=\int_{-\infty}^{\infty}\int_{\Gamma_t}\varphi(x,t)\,d\Gamma_t\,dt
=\int_{\Gamma}\varphi(x,t)\frac{\,d\Gamma}{\sqrt{1+G^2}},
\quad \forall \, \varphi\in {\cD}(\bR^n\times\bR).
$$
\end{definition}

In view of (\ref{g-4.0-10*}), the $\delta$-derivative in (\ref{g-4.0-10-en})
can be rewritten as the Lagrangian derivative:
$$
\frac{\delta\varphi}{\delta t}=\frac{\partial \varphi}{\partial t}
+G\frac{\partial\varphi}{\partial\nu}
=\frac{\partial \varphi}{\partial t}+U_{\delta}\cdot\nabla \varphi
=\frac{D\varphi}{Dt}.
$$

\subsection{Rankine--Hugoniot conditions.}\label{s2.2}
Using Definition~\ref{g-de3-1-en}, we derive the {\em $\delta$-shock
Rankine--Hugoniot conditions} for system (\ref{g-2-en}).

\begin{theorem}
\label{th1-10-en}
Let us assume that $\Omega\subset \bR^n\times (0, \infty)$
is a region cut by a smooth hypersurface $\Gamma=\bigl\{(x,t): S(x,t)=0\bigr\}$
into left- and right-hand parts $\Omega^{\mp}=\{(x,t): \mp S(x,t)>0\}$.
Let $(U,\rho,H)$, $\Gamma$ be a $\delta$-shock wave type solution of system
{\rm (\ref{g-2-en})} {\rm(}in the sense of Definition~{\rm\ref{g-de3-1-en})},
and suppose that $U,\rho,H$ are smooth in $\Omega^{\pm}$ and have one-sided
limits $U^{\pm}$, ${\widehat \rho}^{\pm}$, $H^{\pm}$ on $\Gamma$. Then the
{\em Rankine--Hugoniot conditions for the $\delta$-shock}
\begin{equation}
\label{g-51-10*-en}
\begin{array}{rcl}
\displaystyle
\frac{\delta e}{\delta t}+\nabla_{\Gamma_t}\cdot(eU_{\delta})&=&
\displaystyle
\bigl([\rho U]-[\rho]U_{\delta}\bigr)\cdot\nu, \\
\displaystyle
\frac{\delta (eU_{\delta})}{\delta t}
+\nabla_{\Gamma_t}\cdot(eU_{\delta}\otimes U_{\delta})&=&
\displaystyle
\bigl([\rho U\otimes U]-[\rho U]U_{\delta}\bigr)\cdot\nu, \medskip \\
\displaystyle
\frac{\delta}{\delta t}\Big(\frac{e|U_{\delta}|^2}{2}+h\Big) \qquad\qquad\qquad && \\
\displaystyle
\qquad\qquad
+\nabla_{\Gamma_t}\cdot\Big(\Big(\frac{e|U_{\delta}|^2}{2}+h\Big)U_{\delta}\Big)&=&
\displaystyle
\biggl(\Big[\Big(\frac{\rho |U|^2}{2}+H\Big)U\Big] \medskip \\
\displaystyle
&&\qquad\quad
\displaystyle
-\Big[\frac{\rho |U|^2}{2}+H\Big]U_{\delta}\biggr)\cdot\nu, \\
\end{array}
\end{equation}
hold on the discontinuity hypersurface $\Gamma$, where
$\bigl[f(U,\rho,H)\bigl]=f(U^{-},\rho^{-},H^{-})-f(U^{+},\rho^{+},H^{+})$
is the jump of the function $f(U,\rho,H)$ across the
discontinuity hypersurface $\Gamma$, \
$\frac{\delta}{\delta t}$ is the $\delta$-derivative {\rm(\ref{g-74})}
with respect to $t$, and $\nabla_{\Gamma_t}$ is defined by {\rm(\ref{g-74})},
{\rm(\ref{g-74-3})}.
\end{theorem}

\begin{proof}
The first two conditions in (\ref{g-51-10*-en}) were proved in~\cite[Theorem~9.1.]{umn}.

Let us prove the third condition in (\ref{g-51-10*-en}).
For any test function $\varphi \in {\cD}(\Omega)$ we have
$\varphi(x,t)=0$ for $(x,t)\not\in G$, ${\overline G}\subset \Omega$.
Selecting the test function $\varphi(x,t)$ with compact support in
$\Omega^{\pm}$, we deduce from the third identity in (\ref{g-4.0-10-en})
that the third relation in (\ref{g-2-en}) hold in $\Omega^{\pm}$, i.e.,
\begin{equation}
\label{g-51-10*-en-11}
\big(\frac{\rho|U|^2}{2}+H\big)_t
+\nabla\cdot\big(\big(\frac{\rho|U|^2}{2}+H\big)U\big)=0 \quad \text{for} \quad (x,t)\in \Omega^{\pm}.
\end{equation}

Now, if the test function $\varphi(x,t)$ has the support in $\Omega$, then
$$
\int_{0}^{\infty}\int \bigg(\Big(\frac{{\widehat \rho}|U|^2}{2}+{\widehat H}\Big)\varphi_{t}
+\Big(\frac{{\widehat \rho}|U|^2}{2}+{\widehat H}\Big)U\cdot\nabla\varphi\bigg)\,dx\,dt
\qquad\qquad\qquad\qquad\qquad
$$
$$
=\int_{\Omega^{-}\cap G}
\bigg(\Big(\frac{{\widehat \rho}|U|^2}{2}+{\widehat H}\Big)\varphi_{t}
+\Big(\frac{{\widehat \rho}|U|^2}{2}+{\widehat H}\Big)U\cdot\nabla\varphi\bigg)\,dx\,dt
\qquad\qquad\qquad
$$
$$
\qquad\qquad\quad
+\int_{\Omega^{+}\cap G}
\bigg(\Big(\frac{{\widehat \rho}|U|^2}{2}+{\widehat H}\Big)\varphi_{t}
+\Big(\frac{{\widehat \rho}|U|^2}{2}+{\widehat H}\Big)U\cdot\nabla\varphi\bigg)\,dx\,dt.
$$
Using the integrating-by-parts formula, we obtain
$$
\int\limits_{\Omega^{\pm}\cap G}
\bigg(\Big(\frac{{\widehat \rho}|U|^2}{2}+{\widehat H}\Big)\varphi_{t}
+\Big(\frac{{\widehat \rho}|U|^2}{2}+{\widehat H}\Big)U\cdot\nabla\varphi\bigg)\,dx\,dt
\qquad\qquad\qquad\qquad\qquad\qquad
$$
$$
=-\int\limits_{\Omega^{\pm}\cap G}
\bigg(
\Big(\frac{\rho|U|^2}{2}+H\Big)_t
+\nabla\cdot\Big(\Big(\frac{\rho|U|^2}{2}+H\Big)U\Big)\bigg)\varphi(x,t)\,dx\,dt
$$
$$
\mp\int\limits_{\Gamma\cap G}\bigg(\Big(\frac{\rho^{\pm}|U^{\pm}|^2}{2}+H^{\pm}\Big)\frac{S_t}{|\nabla_{(x,t)}S|}
+\Big(\frac{\rho^{\pm}|U^{\pm}|^2}{2}+H^{\pm}\Big)\frac{U^{\pm}\cdot\nabla S}{|\nabla_{(x,t)}S|}\bigg)
\varphi(x,t)\,d\Gamma
$$
$$
-\int\limits_{\Omega^{\pm}\cap G\cap \bR^n}\Big(\frac{{\widehat \rho}^0(x)|U^0(x)|^2}{2}+{\widehat H}^0(x)\Big)
\varphi(x,0)\,dx,
$$
where $d\Gamma$ is the surface measure on $\Gamma$.
Next, adding the latter relations and taking into account (\ref{g-51-10*-en-11}), we have
$$
\int_{0}^{\infty}\int \bigg(\Big(\frac{{\widehat \rho}|U|^2}{2}+{\widehat H}\Big)\varphi_{t}
+\Big(\frac{{\widehat \rho}|U|^2}{2}+{\widehat H}\Big)U\cdot\nabla\varphi\bigg)\,dx\,dt
\qquad\qquad\qquad\qquad\qquad
$$
$$
+\int\Big(\frac{{\widehat \rho}^0(x)|U^0(x)|^2}{2}+{\widehat H}^0(x)\Big)\varphi(x,0)\,dx
$$
\begin{equation}
\label{97-10-en}
=\int_{\Gamma}\Big(-\Bigl[\frac{\rho|U|^2}{2}+H\Bigr]G+\Bigl[\Big(\frac{\rho|U|^2}{2}+H\Big)U\Bigr]\cdot\nu\Big)
\varphi(x,t)\frac{\,d\Gamma}{\sqrt{1+G^2}}.
\end{equation}
Next, applying the integrating-by-parts formula (\ref{g-84.20}) to the second summand
in third identity (\ref{g-4.0-10-en}), one can see that
$$
\int_{\Gamma}\Big(\frac{e|U_{\delta}|^2}{2}+h\Big)
\frac{\delta\varphi}{\delta t}\frac{\,d\Gamma}{\sqrt{1+G^2}}
+\int_{\Gamma_0}\Big(\frac{e^0(x)|U_{\delta}^0(x)|^2}{2}+h^0(x)\Big)\varphi(x,0)\,d\Gamma_{0}
$$
$$
\qquad\qquad\qquad\qquad
=-\int_{\Gamma}\frac{\delta^*}{\delta t}\Big(\frac{e|U_{\delta}|^2}{2}+h\Big)
\varphi\frac{\,d\Gamma}{\sqrt{1+G^2}},
$$
where the adjoint operator $\frac{\delta^*}{\delta t}$ is defined in
(\ref{g-84.20-1}). Thus
$$
\int_{\Gamma}\Big(\frac{e|U_{\delta}|^2}{2}+h\Big)
\frac{\delta\varphi}{\delta t}\frac{\,d\Gamma}{\sqrt{1+G^2}}
+\int_{\Gamma_0}\Big(\frac{e^0(x)|U_{\delta}^0(x)|^2}{2}+h^0(x)\Big)\varphi(x,0)\,d\Gamma_{0}
$$
\begin{equation}
\label{98-10-en}
=-\int_{\Gamma}\biggl(\frac{\delta}{\delta t}\Big(\frac{e|U_{\delta}|^2}{2}+h\Big)
+\nabla_{\Gamma_t}\cdot\Big(\Big(\frac{e|U_{\delta}|^2}{2}+h\Big)G\nu\Big)\biggr)
\varphi\frac{\,d\Gamma}{\sqrt{1+G^2}}.
\end{equation}

Adding (\ref{97-10-en}) and (\ref{98-10-en}) and taking into account (\ref{g-4.0-10-en}),
(\ref{g-4.0-10*}), we derive
$$
\int_{\Gamma}\biggl(-\Bigl[\frac{\rho|U|^2}{2}+H\Bigr]U_{\delta}\cdot\nu
+\Bigl[\Big(\frac{\rho|U|^2}{2}+H\Big)U\Bigr]\cdot\nu
\qquad\qquad\qquad\qquad\qquad
$$
$$
-\frac{\delta}{\delta t}\Big(\frac{e|U_{\delta}|^2}{2}+h\Big)
-\nabla_{\Gamma_t}\cdot\Big(\Big(\frac{e|U_{\delta}|^2}{2}+h\Big)U_{\delta}\Big)\biggr)
\varphi(x,t)\frac{\,d\Gamma}{\sqrt{1+G^2}}=0,
$$
for all $\varphi \in {\cD}(\Omega)$. Thus, the third relation in
(\ref{g-51-10*-en}) holds.
\end{proof}

The right-hand sides of the equations in (\ref{g-51-10*-en}) are called the
{\em Rankine--Hugoniot deficits} in $\rho$, $\rho U$, and $\frac{\rho|U|^2}{2}+H$,
respectively.

Let $a(x,t)$ be a smooth function defined only on the surface
$\Gamma=\bigl\{(x,t):S(x,t)=0\bigr\}$ which is the restriction of
some smooth function defined in a neighborhood of $\Gamma$ in $\bR^n\times\bR$.
It is easy to prove that
\begin{equation}
\label{g-51-12*-kur} \nabla_{\Gamma_t}\cdot(aU_{\delta})=-2{\cK} Ga,
\end{equation}
where ${\cK}$ is the mean curvature of the surface $\Gamma_t$ (see
{\rm(\ref{g-84.3})} ). Indeed, according to (\ref{g-74}),
(\ref{g-74-3}), (\ref{g-84.3}), (\ref{g-4.0-10*}), we have
$\nabla_{\Gamma_t}\cdot(aU_{\delta})=\sum_{k=1}^n\frac{\delta
(Ga\nu_k)}{\delta x_k} =\sum_{k=1}^n\frac{\delta (Ga)}{\delta
x_k}\nu_k+Ga\sum_{k=1}^n\frac{\delta \nu_k}{\delta x_k} =-2{\cK}
Ga$. Here the obvious relation $\sum_{k=1}^n\frac{\delta
(Ga)}{\delta x_k}\nu_k=0$ was taken into account.

Due to (\ref{g-51-12*-kur}), the Rankine--Hugoniot conditions (\ref{g-51-10*-en})
can also be rewritten as
\begin{equation}
\label{g-51-10*-1}
\begin{array}{rcl}
\displaystyle
\frac{\delta e}{\delta t}-2{\cK} Ge&=&
\displaystyle
\bigl([\rho U]-[\rho]U_{\delta}\bigr)\cdot\nu, \\
\displaystyle
\frac{\delta (eU_{\delta})}{\delta t}-2{\cK} GeU_{\delta}&=&
\displaystyle
\bigl([\rho U\otimes U]-[\rho U]U_{\delta}\bigr)\cdot\nu, \\
\displaystyle
\frac{\delta}{\delta t}\Big(\frac{e|U_{\delta}|^2}{2}+h\Big)
-2{\cK} G\Big(\frac{e|U_{\delta}|^2}{2}+h\Big)&=&
\displaystyle
\biggl(\Big[\Big(\frac{\rho|U|^2}{2}+H\Big)U\Big] \\
&&\qquad\quad
\displaystyle
-\Big[\frac{\rho|U|^2}{2}+H\Big]U_{\delta}\biggr)\cdot\nu. \\
\end{array}
\end{equation}

\begin{remark}
\label{rem1-10} \rm
The Rankine--Hugoniot conditions (\ref{g-51-10*-en}) constitute
a system of {\em second-order} PDEs.
According to this fact, for system
(\ref{g-2-en}) we use the initial data (\ref{5.0-10-en})
which contain the {\em initial velocity} $U^{0}_{\delta}(x)$ of a $\delta$-shock.
This is similar to the fact that in the {\em measure-valued solution}
approach~\cite{B},~\cite{Li-Zh1},~\cite{Li-Y},~\cite{Y-1} the velocity $U$ is
determined on the discontinuity surface.
\end{remark}

In the direction $\nu$ the characteristic equation of system (\ref{g-2-en})
has repeated eigenvalues $\lambda=U\cdot\nu$.
So, we assume that for the initial data (\ref{5.0-10-en}) the
{\em geometric entropy condition} holds:
\begin{equation}
\label{g-4}
U^{0+}(x)\cdot \nu^0\bigr|_{\Gamma_0}
< U_{\delta}^0(x)\cdot\nu^0\bigr|_{\Gamma_0}
< U^{0-}(x)\cdot\nu^0\bigr|_{\Gamma_0},
\end{equation}
where $\nu^0=\frac{\nabla S^0(x)}{|\nabla S^0(x)|}$ is the unit
space normal of $\Gamma_0$, oriented from $\Omega^{-}_0=\{x\in \bR^{n}:S^0(x)<0\}$
to $\Omega^{+}_0=\{x\in \bR^{n}: S^0(x)>0\}$.
Similarly, we assume that for a solution of the Cauchy problem (\ref{g-2-en}),
(\ref{5.0-10-en}) the {\em geometric entropy condition} holds:
\begin{equation}
\label{g-55}
U^{+}(x,t)\cdot \nu\bigr|_{\Gamma_t}
< U_{\delta}(x,t)\cdot \nu\bigr|_{\Gamma_t}
< U^{-}(x,t)\cdot \nu\bigr|_{\Gamma_t},
\end{equation}
where $U_{\delta}$ is the velocity (\ref{g-4.0-10*}) of the
$\delta$-shock front $\Gamma_t$, $U^{\pm}$ is the velocity behind the
$\delta$-shock wave front and ahead of it, respectively.
Condition (\ref{g-55}) implies that all characteristics on both sides
of the discontinuity $\Gamma_t$ must overlap. For $t=0$ the condition
(\ref{g-55}) coincides with (\ref{g-4}).

\section{$\delta$-shock mass, momentum and energy transport relations}
\label{s3}

The classical conservation laws (\ref{7.0}) do not make sense  for a
$\delta$-shock wave type solution.
 ``Generalized'' analogs of conservation laws
(\ref{7.0}) were derived in~\cite{Al-S},~\cite{Pan-S1},~\cite{S8} for the
one-dimensional case, and in~\cite{umn} for the multidimensional case.
Now we derive these transport conservation laws for the case of system (\ref{g-2-en}).

Let us assume that a moving surface $\Gamma_{t}=\bigl\{x: S(x,t)=0\bigr\}$
permanently separates $\bR^{n}_x$ into two parts $\Omega^{\pm}_{t}=\{x\in \bR^{n}: \pm S(x,t)>0\}$,
and $\Omega^{\pm}_{0}=\{x\in \bR^{n}: \pm S^0(x)>0\}$. Let $(U, \rho, H)$ be
compactly supported with respect to~$x$. Denote by
\begin{equation}
\label{g-61-1}
M(t)=\int_{\Omega^{-}_{t}\cup\Omega^{+}_{t}}\rho(x,t)\,dx,
\quad m(t)=\int_{\Gamma_{t}}e(x,t)\,d\Gamma_{t},
\end{equation}
and
\begin{equation}
\label{g-61-2}
P(t)=\int_{\Omega^{-}_{t}\cup\Omega^{+}_{t}}\rho(x,t)U(x,t)\,dx,
\quad
p(t)=\int_{\Gamma_{t}}e(x,t)U_{\delta}(x,t)\,d\Gamma_{t},
\end{equation}
masses and momenta of the volume $\Omega^{-}_{t}\cup\Omega^{+}_{t}$ and
the moving $\delta$-shock wave front $\Gamma_{t}$, respectively, $d\Gamma_{t}$
being the surface measure on $\Gamma_{t}$.
Let
\begin{equation}
\label{g-61-2-1}
W_{kin}(t)=\int\limits_{\Omega^{-}_{t}\cup\Omega^{+}_{t}}\frac{\rho(x,t)|U(x,t)|^2}{2}\,dx,
\quad
w_{kin}(t)=\int\limits_{\Gamma_{t}}\frac{e(x,t)|U_{\delta}(x,t)|^2}{2}\,d\Gamma_{t},
\end{equation}
and
\begin{equation}
\label{g-61-2-2}
W_{int}(t)=\int_{\Omega^{-}_{t}\cup\Omega^{+}_{t}}H(x,t)\,dx,
\quad
w_{int}(t)=\int_{\Gamma_{t}}h(x,t)\,d\Gamma_{t},
\end{equation}
be the kinetic and internal energies of the volume $\Omega^{-}_{t}\cup\Omega^{+}_{t}$
and the moving wave front $\Gamma_{t}$, respectively.
Here $W_{kin}(t)+w_{kin}(t)$ and $W_{int}(t)+w_{int}(t)$ are the total kinetic and
internal energies, respectively; $W_{kin}(t)+w_{kin}(t)+W_{int}(t)+w_{int}(t)$ is
the total energy.

\begin{theorem}
\label{g-th5-en}
Let $(U, \rho, H)$ together with a discontinuity
hypersurface $\Gamma=\bigl\{(x,t): S(x,t)=0\bigr\}$ be a $\delta$-shock wave
type solution {\rm(}in the sense of Definition~{\rm\ref{g-de3-1-en})}
of the Cauchy problem {\rm (\ref{g-2-en})}, {\rm (\ref{5.0-10-en})}, where
$$
\rho(x,t)={\widehat \rho}(x,t)+e(x,t)\delta(\Gamma),
\qquad
H(x,t)={\widehat H}(x,t)+h(x,t)\delta(\Gamma).
$$
Let this solution satisfy the entropy condition {\rm (\ref{g-55})}.
Suppose that $(U, \rho, H)$ is compactly supported with respect to~$x$,
smooth in $\Omega^{\pm}=\{(x,t): \pm S(x,t)>0\}$ and has one-sided limits
$U^{\pm}$, ${\widehat \rho}^{\pm}$, ${\widehat H}^{\pm}$ on $\Gamma$.
Then the following {\em mass and momentum balance relations} hold:
\begin{equation}
\label{g-61*}
\begin{array}{rcl}
\displaystyle
\dot M(t)=-\dot m(t), \,\,\,\,\quad \dot m(t)\geq 0, &&
\qquad
\dot P(t)=-\dot p(t), \smallskip \\
\displaystyle
M(t)+m(t)=M(0)+m(0), && \qquad P(t)+p(t)=P(0)+p(0).
\end{array}
\end{equation}
\end{theorem}

In fact, the proof of Theorem~\ref{g-th5-en} coincides with the proof
of~\cite[Theorem~9.2.]{umn}.
The proof of~\cite[Theorem~9.2.]{umn}, and, consequently, the proof of Theorem~\ref{g-th5-en}
are based on the volume and surface transport Theorems~\ref{g-th8},~\ref{g-th4-tr}
and use the first two relations in (\ref{g-51-12*-kur}).

\begin{theorem}
\label{e-th4-en}
Let $(U, \rho, H)$ together with a discontinuity hypersurface
$\Gamma=\bigl\{(x,t):\, S(x,t)=0\bigr\}$ satisfy the same conditions as in
Theorem~{\rm\ref{g-th5-en}}.
Then the following {\em energy balance relations} hold:
\begin{equation}
\label{e-61-en}
\dot w_{kin}(t)+\dot w_{int}(t)\geq 0,
\quad
\dot W_{kin}(t) \leq 0,
\quad
\dot W_{int}(t)+\dot w_{int}(t)\geq 0,
\quad \dot W_{int}(t) \leq 0.
\end{equation}
Moreover,
\begin{equation}
\label{e-61-en-11}
\begin{array}{rcl}
\displaystyle
\dot W_{kin}(t)+\dot w_{kin}(t)=-\bigl(\dot W_{int}(t)+\dot w_{int}(t)\bigr), \qquad\qquad && \smallskip \\
\displaystyle
W_{kin}(t)+w_{kin}(t)+W_{int}(t)+w_{int}(t)= \qquad\qquad\qquad\qquad && \\
\displaystyle
=W_{kin}(0)+w_{kin}(0)+W_{int}(0)+w_{int}(0).  && \\
\end{array}
\end{equation}
\end{theorem}

\begin{proof}
1. Let us assume that the supports of $U(x,t)$ and $\rho(x,t)$
with respect to~$x$ belong to a compact $K\in\bR^n_x$ bounded by $\partial K$.
Let $K^{\pm}_t=\Omega^{\pm}_t\cap K$.
By $\nu$ we denote, as before, the space normal to $\Gamma_{t}$ pointing from
$\Omega^{-}_t$ to $\Omega^{+}_t$.
Differentiating $W_{kin}(t)+W_{int}(t)$ and using the volume transport Theorem~\ref{g-th8}, we obtain
$$
\dot W_{kin}(t)+\dot W_{int}(t)
=\int_{K^{-}_{t}\cup K^{+}_{t}}\frac{\partial }{\partial t}\Bigl(\frac{\rho(x,t)|U(x,t)|^2}{2}+H(x,t)\Bigr)\,dx
\qquad\qquad
$$
\begin{equation}
\label{e-62-en}
+\int_{\partial K^{-}_{t}\cup\partial K^{+}_{t}}\Bigl(\frac{\rho(x,t)|U(x,t)|^2}{2}+H(x,t)\Bigr)
V(x,t)\cdot\tilde{\nu}\,d\Gamma_{t},
\end{equation}
where $\tilde{\nu}$ is the outward unit space normal to the surface $\partial K^{\pm}_{t}$
and $V(x,t)$ is the velocity of the point $x$ in $K^{\pm}_{t}$.

Next, taking into account that for $x\in K^{\pm}_t$ system (\ref{g-2-en}) has a smooth
solution $(U^{\pm},\rho^{\pm},H^{\pm})$, i.e.,
$$
\Big(\frac{\rho^{\pm}|U^{\pm}|^2}{2}+H^{\pm}\Big)_t
+\nabla\cdot\Big(\Big(\frac{\rho^{\pm}\,|U^{\pm}|^2}{2}+H^{\pm}\Big)U^{\pm}\Big)=0,
$$
and $U^{\pm}$, $\rho^{\pm}$, $H^{\pm}$ are equal to zero on the hypersurface
$\partial K^{\pm}_{t}$ except $\Gamma_{t}$, applying Gauss's divergence theorem to
relation (\ref{e-62-en}), we transform it to the form
$$
\dot W_{kin}(t)+\dot W_{int}(t)
=-\int_{K^{-}_{t}}\nabla\cdot\Big(\Big(\frac{\rho^{-}\,|U^{-}|^2}{2}+H^{-}\Big)U^{-}\Big)\,dx
\qquad\qquad\qquad\qquad
$$
$$
\qquad\quad
-\int_{K^{+}_{t}}\nabla\cdot\Big(\Big(\frac{\rho^{+}\,|U^{+}|^2}{2}+H^{+}\Big)U^{+}\Big)\,dx
+\int_{\Gamma_{t}}\Bigl[\frac{\rho \,|U|^2}{2}+H\Bigr]U_{\delta}\cdot\nu\,d\Gamma_{t}
$$
$$
=-\int_{\Gamma_{t}}\Big(\frac{\rho^{-}\,|U^{-}|^2}{2}+H^{-}\Big)U^{-}\cdot\nu \,d\Gamma_{t}
\qquad\qquad\qquad\qquad\qquad\qquad
$$
$$
\qquad
+\int_{\Gamma_{t}}\Big(\frac{\rho^{+}\,|U^{+}|^2}{2}+H^{+}\Big)U^{+}\cdot\nu\ \,d\Gamma_{t}
+\int_{\Gamma_{t}}\Bigl[\frac{\rho \,|U|^2}{2}+H\Bigr]U_{\delta}\cdot\nu\,d\Gamma_{t}
$$
\begin{equation}
\label{e-63-en}
=-\int_{\Gamma_{t}}\biggl(\Bigl[\Big(\frac{\rho\,|U|^2}{2}+H\Big)U\Bigr]
-\Bigl[\frac{\rho \,|U|^2}{2}+H\Bigr]U_{\delta}\biggr)\cdot\nu\,d\Gamma_{t},
\qquad\qquad
\end{equation}
where $U_{\delta}=V\bigl|_{\Gamma_{t}}$ is the velocity (\ref{g-4.0-10*}) of the
$\delta$-shock front $\Gamma_t$.
Using the third Rankine--Hugoniot condition (\ref{g-51-10*-en}), relation
(\ref{e-63-en}) can be rewritten as
$$
\dot W_{kin}(t)+\dot W_{int}(t)
\qquad\qquad\qquad\qquad\qquad\qquad\qquad\qquad\qquad\qquad\qquad\qquad
$$
\begin{equation}
\label{e-63-en-1}
=-\int_{\Gamma_{t}}\biggl(\frac{\delta}{\delta t}\Big(\frac{e|U_{\delta}|^2}{2}+h\Big)
+\nabla_{\Gamma_t}\cdot\Big(\Big(\frac{e|U_{\delta}|^2}{2}+h\Big)U_{\delta}\Big)\biggr)\,d\Gamma_{t}.
\end{equation}
Applying the surface transport Theorem~\ref{g-th4-tr} to the second relations in
(\ref{g-61-2-1}), (\ref{g-61-2-2}) one can see that the right-hand side of
(\ref{e-63-en-1}) coincides with $-\dot w_{kin}(t)-\dot w_{int}(t)$.
Thus relations (\ref{e-61-en-11}) hold.

Since $\rho^{\pm}\ge 0$, $H^{\pm}\ge 0$ and the solution $(U, \rho, H)$ of the Cauchy problem
(\ref{g-2-en}), (\ref{5.0-10-en}) satisfies the entropy condition (\ref{g-55}), we have
$$
\big([\rho |U|^2U]-[\rho |U|^2]U_{\delta}\big)\cdot\nu
\qquad\qquad\qquad\qquad\qquad\qquad\qquad\qquad\qquad\qquad
$$
\begin{equation}
\label{e-63-en-2}
=\big(\rho^{-}|U^{-}|^2(U^{-}-U_{\delta})\cdot\nu
+\rho^{+}|U^{+}|^2(U_{\delta}-U^{+})\cdot\nu\big)\bigr|_{\Gamma_t}\ge 0;
\end{equation}
\begin{equation}
\label{e-63-en-3}
\big([HU]-[H]U_{\delta}\big)\cdot\nu=\big(H^{-}(U^{-}-U_{\delta})\cdot\nu
+H^{+}(U_{\delta}-U^{+})\cdot\nu\big)\bigr|_{\Gamma_t}\ge 0.
\end{equation}
Formulas (\ref{e-63-en}), (\ref{e-63-en-2}), (\ref{e-63-en-3}) imply that
$\dot W_{kin}(t)+\dot W_{int}(t)\leq 0$, i.e., due to (\ref{e-61-en-11})
the first inequality in (\ref{e-61-en}) holds.

2. In fact, the second inequality in (\ref{e-61-en}) was proved in~\cite{S-2008}.

Let us calculate $\dot w(t)$.
Taking into account formula (\ref{g-51-12*-kur}), due to the surface transport
Theorem~\ref{g-th4-tr}, we obtain
$$
\dot w_{kin}(t)
=\frac{1}{2}\int_{\Gamma_{t}}\Big(\frac{\delta}{\delta t}\big(e(x,t)|U_{\delta}(x,t)|^2\big)
+\nabla_{\Gamma_t}\cdot(e(x,t)|U_{\delta}(x,t)|^2U_{\delta})\Big)\,d\Gamma_{t}
$$
$$
=\frac{1}{2}\int_{\Gamma_{t}}\Big(\frac{\delta}{\delta t}\big(e(x,t)|U_{\delta}(x,t)|^2\big)
-2{\cK} Ge(x,t)|U_{\delta}(x,t)|^2\Big)\,d\Gamma_{t}
$$
\begin{equation}
\label{e-64}
=\frac{1}{2}\int_{\Gamma_{t}}
\Big(\sum_{k=1}^n\Big(u_{\delta k}\frac{\delta(eu_{\delta k})}{\delta t}
+u_{\delta k}e\frac{\delta u_{\delta k}}{\delta t}\Big)
-2{\cK} Ge(x,t)|U_{\delta}(x,t)|^2\Big)\,d\Gamma_{t}.
\end{equation}

According to (\ref{g-51-10*-1}) and (\ref{g-51-12*-kur}), we have
\begin{equation}
\label{e-64-1}
\begin{array}{rcl}
\displaystyle
\frac{\delta e}{\delta t}\,u_{\delta k}+e\frac{\delta u_{\delta k}}{\delta t}
-2{\cK} Ge\,u_{\delta k}&=&[\rho u_{k}U\cdot\nu]-[\rho u_{k}]U_{\delta}\cdot\nu, \\
\displaystyle
\frac{\delta e}{\delta t}\,u_{\delta k}
-2{\cK} Geu_{\delta k}&=&[\rho U\cdot\nu]\,u_{\delta k}-[\rho]U_{\delta}\cdot\nu \,u_{\delta k}, \\
\end{array}
\end{equation}
where $u_{\delta k}(x,t)$ is the $k$-th component of the vector $U_{\delta}$, $k=1,\dots,n$.
Now, subtracting one equation from the other in (\ref{e-64-1}), we obtain
\begin{equation}
\label{e-64-2}
e\frac{\delta u_{\delta k}}{\delta t}
=[\rho u_{k}U\cdot\nu]-[\rho u_{k}]U_{\delta}\cdot\nu
-[\rho U\cdot\nu]u_{\delta k}+[\rho]U_{\delta}\cdot\nu u_{\delta k}.
\end{equation}
Substituting equations (\ref{e-64-2}) into (\ref{e-64}), one can easily calculate
$$
\dot w_{kin}(t)
=\frac{1}{2}\int_{\Gamma_{t}}
\Big(2\sum_{k=1}^n\big([\rho u_{k}U\cdot\nu]-[\rho u_{k}]U_{\delta}\cdot\nu\big)u_{\delta k}
\qquad\qquad\qquad\qquad
$$
$$
\qquad\qquad
-[\rho U\cdot\nu]|U_{\delta}|^2+[\rho]|U_{\delta}|^2U_{\delta}\cdot\nu\Big)\,d\Gamma_{t}.
$$
Taking into account that $U_{\delta}=G\nu$, $G=-\frac{S_{t}}{|\nabla S|}$, i.e.,
$u_{\delta k}=G\nu_k$, $k=1,\dots,n$, we rewrite the above
relation as
\begin{equation}
\label{e-65}
\dot w_{kin}(t)=\frac{1}{2}\int_{\Gamma_{t}}
\Big(2[\rho (U\cdot\nu)^2]G-3[\rho U\cdot\nu]G^2+[\rho]G^3\Big)\,d\Gamma_{t}.
\end{equation}

Using (\ref{e-63-en}) and (\ref{e-65}), we obtain
$$
\dot W_{kin}(t)+\dot w_{kin}(t)=-\frac{1}{2}\int_{\Gamma_{t}}
\big([\rho |U|^2U\cdot\nu]-[\rho |U|^2]U_{\delta}\cdot\nu
\qquad\qquad\qquad\qquad
$$
\begin{equation}
\label{e-66-en}
\qquad
-2[\rho (U\cdot\nu)^2]G+3[\rho U\cdot\nu]G^2-[\rho]G^3\big)\,d\Gamma_{t}.
\end{equation}
The gas velocity $U|_{\Gamma_{t}}$ on the wave front $\Gamma_{t}$ is
the sum of the normal component $U\cdot\nu$ and the component $U_{tan}$ tangential
to the surface $\Gamma_{t}$. Since $|U|^2|_{\Gamma_{t}}=(U\cdot\nu)^2+U_{tan}^2$,
and $G=U_{\delta}\cdot\nu$, one can represent the integrand in (\ref{e-66-en}) as
$$
\big([\rho |U|^2U\cdot\nu]-[\rho |U|^2]U_{\delta}\cdot\nu
-2[\rho (U\cdot\nu)^2]G+3[\rho U\cdot\nu]G^2-[\rho]G^3\big)
\qquad\qquad
$$
$$
\qquad
=\rho^{-}(U_{tan}^{-})^2(U^{-}\cdot\nu-U_{\delta}\cdot\nu)
+\rho^{+}(U_{tan}^{+})^2(U_{\delta}\cdot\nu-U^{+}\cdot\nu)
$$
\begin{equation}
\label{e-67-en}
\qquad\qquad
+\rho^{-}(U^{-}\cdot\nu-U_{\delta}\cdot\nu)^3
+\rho^{+}(U_{\delta}\cdot\nu-U^{-}\cdot\nu)^3.
\end{equation}
Since a solution $(U, \rho, H)$ of the Cauchy problem (\ref{g-2-en}), (\ref{5.0-10-en})
satisfies the entropy condition (\ref{g-55}) and $\rho^{\pm}\ge 0$, we deduce
that the expression (\ref{e-67-en}) is {\em non-negative}.
Formulas (\ref{e-66-en}), (\ref{e-67-en}) imply that $\dot W_{kin}(t)+\dot w_{kin}(t)\leq 0$.
Due to (\ref{e-61-en-11}), we conclude that the third inequality in (\ref{e-61-en}) holds.

3. Since $U,\rho, H$ are smooth in $\Omega^{\pm}=\{(x,t): \pm S(x,t)>0\}$, it easy to see that
for $(x,t)\in K^{\pm}_{t}$ the first and second equations in (\ref{g-2-en}) imply that
$$
(u^{\pm}_k)_t+\sum_{j=1}^n u_j\frac{\partial u^{\pm}_k}{\partial x_j}=0, \quad k=1,2,\dots,n.
$$
Multiplying the both sides of the above equation by $u^{\pm}_k$ and summarizing over $k=1,2,\dots,n$, we obtain
\begin{equation}
\label{e-67-int-2}
(|U^{\pm}|^2)_t+\sum_{j=1}^n u^{\pm}_j\frac{\partial (|U|^{\pm})^2}{\partial x_j}=0, \quad (x,t)\in K^{\pm}_{t}.
\end{equation}
According to (\ref{e-67-int-2}) and the first equation in (\ref{g-2-en})
\begin{equation}
\label{e-67-int-3}
\Bigl(\frac{\rho^{\pm}|U^{\pm}|^2}{2}\Bigr)_t
+\sum_{j=1}^n \frac{\partial}{\partial x_j}\Bigl(\frac{\rho^{\pm}|U^{\pm}|^2}{2}u^{\pm}_j\Bigr)=0,
\quad (x,t)\in K^{\pm}_{t}.
\end{equation}
In the end, from (\ref{e-67-int-3}) and the third equation in (\ref{g-2-en}) we obtain that
\begin{equation}
\label{e-67-int-4}
\big(H^{\pm}\big)_t+\nabla\cdot\big(H^{\pm} U^{\pm}\big)=0, \quad (x,t)\in K^{\pm}_{t}.
\end{equation}

Next, as before, differentiating $W_{int}(t)$, using (\ref{e-67-int-4}) and applying
the volume transport Theorem~\ref{g-th8}, we obtain
$$
\dot W_{int}(t)
=\int_{K^{-}_{t}\cup K^{+}_{t}}\frac{\partial H(x,t)}{\partial t}\,dx
+\int_{\partial K^{-}_{t}\cup\partial K^{+}_{t}}H(x,t)V(x,t)\cdot\tilde{\nu}\,d\Gamma_{t},
\qquad\qquad
$$
\begin{equation}
\label{e-67-int-1}
=-\int_{K^{-}_{t}\cup K^{+}_{t}}\nabla\cdot\big(H U\big)\,dx
+\int_{\partial K^{-}_{t}\cup\partial K^{+}_{t}}H(x,t)V(x,t)\cdot\tilde{\nu}\,d\Gamma_{t},
\end{equation}
where $\tilde{\nu}$ is the outward unit space normal to the surface $\partial K^{\pm}_{t}$
and $V(x,t)$ is the velocity of the point $x$ in $K^{\pm}_{t}$.

Taking into account that $U^{\pm}$, $\rho^{\pm}$, $H^{\pm}$ are equal to zero on the hypersurface
$\partial K^{\pm}_{t}$ except $\Gamma_{t}$ and applying Gauss's divergence theorem to
(\ref{e-67-int-1}), we transform this relation to the form
$$
\dot W_{int}(t)
=-\int_{K^{-}_{t}}\nabla\cdot\big(H^{-} U^{-}\big)\,dx
-\int_{K^{+}_{t}}\nabla\cdot\big(H^{+}U^{+}\big)\,dx
+\int_{\Gamma_{t}}\bigl[H\bigr]U_{\delta}\cdot\nu\,d\Gamma_{t}
$$
$$
\qquad
=-\int_{\Gamma_{t}}H^{-} U^{-}\cdot\nu \,d\Gamma_{t}
+\int_{\Gamma_{t}}H^{+} U^{+}\cdot\nu\ \,d\Gamma_{t}
+\int_{\Gamma_{t}}\bigl[H\bigr]U_{\delta}\cdot\nu\,d\Gamma_{t}
$$
\begin{equation}
\label{e-67-int-5}
=-\int_{\Gamma_{t}}\bigl(\bigl[H U\bigr]-\bigl[H\bigr]U_{\delta}\bigr)\cdot\nu\,d\Gamma_{t},
\qquad\qquad\qquad\qquad\qquad
\end{equation}
where $U_{\delta}=V\bigl|_{\Gamma_{t}}$ is the velocity (\ref{g-4.0-10*}) of the
$\delta$-shock front $\Gamma_t$, $\nu=\tilde{\nu}\bigl|_{\Gamma_{t}}$ is the space normal
to $\Gamma_{t}$ pointing from $K^{-}_t$ to $K^{+}_t$.
In view of the entropy condition (\ref{g-55}) the inequality  (\ref{e-63-en-3}) holds, and
consequently, (\ref{e-67-int-5}) implies the fourth inequality in (\ref{e-61-en}).
\end{proof}

\begin{corollary}
\label{cor-fin}
According to Theorems~{\rm\ref{g-th5-en},~\ref{e-th4-en}}, the mass, momentum
and energy {\em transport processes} between the volume outside of the $\delta$-shock
wave front $\Omega^{-}_{t}\cup \Omega^{+}_{t}=\{x\in \bR^{n}: S(x,t)\ne 0\}$
and the moving $\delta$-shock wave front $\Gamma_t$ {\em are going on} such that the total
mass $M(t)+m(t)$, momentum $P(t)+p(t)$ and energy $W_{kin}(t)+w_{kin}(t)+W_{int}(t)+w_{int}(t)$
are independent of time. More precisely the {\em mass and energy concentration processes}
on the moving $\delta$-shock wave front $\Gamma_t$ {\em are going on}.
In addition, the total kinetic energy $W_{kin}(t)+w_{kin}(t)$ transfers into the total
internal energy $W_{int}(t)+w_{int}(t)$.
\end{corollary}

The inequality $\dot W(t)\leq 0$ in (\ref{e-61-en}) reflects the well-known
fact that the evolution of a solution with shocks is connected with decreasing
of the kinetic energy.

\begin{remark}
\label{rem1-11} \rm
Let us suppose that in a {\it finite time} $\tilde{t}$ the whole initial mass $M(0)$ and
energy $W_{kin}(0)+W_{int}(0)$ may be concentrated on the $\delta$-shock front $\Gamma_{t}$.
Then, according to The Rankine--Hugoniot conditions, for $t>\tilde{t}$,
{\it instead of the whole system} of zero-pressure gas dynamics (\ref{g-2-en}) we obtain
exactly  a {\it ``surface'' version of this system}
$$
\frac{\delta e}{\delta t}+\nabla_{\Gamma}\cdot(eU_{\delta})=0,
$$
$$
\frac{\delta (eU_{\delta})}{\delta t}
 +\nabla_{\Gamma}\cdot(eU_{\delta}\otimes U_{\delta})=0,
$$
$$
\frac{\delta}{\delta t}\Big(\frac{e|U_{\delta}|^2}{2}+h\Big)
+\nabla_{\Gamma_t}\cdot\Big(\Big(\frac{e|U_{\delta}|^2}{2}+h\Big)U_{\delta}\Big)=0,
$$
where $U_{\delta}$ is the velocity of the $\delta$-shock front
$\Gamma_t$, $e$ is the surface density of the front mass,
$h$ is the surface density of the front internal energy.
This system is an analog of the initial system of zero-pressure gas dynamics (\ref{g-2-en})
on the $(n-1)$-- dimensional surface $\Gamma_{t}$.
This $(n-1)$ -- dimensional analog also has the same type as the initial
system, therefore its solution can develop singularities
within a finite time, and all mass concentrates on the manifold of
dimension $n-2$, and so on. Thus, it may happen that after the finite
number of bifurcations the whole initial mass will be
concentrated at the singular point.

\end{remark}

\section{Example of an one dimensional concentration process}
\label{s4}

In the 1D case we construct  an explicit example of the
concentration process based on another method. Namely, let us
consider the data that do not imply the $\delta$-shock initially:
\begin{equation}
\begin{array}{c}
\label{E1}
U^0(x)=(U^-- [U]\Theta(x))\,\chi_I(x)), \\
\rho^0(x)=(\rho^-- [\rho]\Theta(x))\,\chi_I(x), \\
H^0(x)=(H^-- [H]\Theta(x))\,\chi_I(x)),
\end{array}
\end{equation}
where $ U^-, \rho^-, H^-, U^+, \rho^+, H^+$ are constants, $\rho^-,
\rho^+, H^-, H^+>0$, $[U]>0$, $\chi_I(x)$ is the characteristic
function of the segment $I=[-L, L], L\gg 1.$ Let us note that we can
apply the standard mollification procedure to obtain  functions
smooth at the points $\pm L$, but here do not need to do it.

We obtain the solution to the Cauchy problem by means of the free
particles method~\cite{AKR}: first we assume that the particles do
not feel one others and form the overlapping domain. Then we switch
to the sticky particles model and change this overlapping domain to
a point where the mass accumulates according to the conservation of
mass and momentum. Now we have to consider the additional law of
conservation of energy. Thus, according to (\cite{AKR}), the
free-particles solution $(\rho_{FP}, U_{FP})$ to the two first
equations to the zero pressure model has the form
$$
\rho_{FP}(t,x)=\begin{cases}
\rho^-,&\text{$-L+(U^--[U])t<x<(U^--[U])t$,}\\
2\rho^--[\rho],&\text{$(U^--[U])t<x<U_-t$,}\\
\rho_--[\rho],&\text{$L+U^-t>x>U^-t$,}\\
0,& \text{otherwise},
\end{cases}
$$
$$
U_{FP}(t,x)=\begin{cases}
U^-,&\text{$-L+(U^--[U])t<x<(U^--[U])t$,}\\
U^-+\dfrac{\rho_--[\rho]}{2\rho^--[\rho]}[U],&\text{$(U^--[U])t<x<U^-t$,}\\
U^--[U],&\text{$L+U^-t>x>U^-t$.}\\
\end{cases}
$$
Outside of the segment $[-L+(U^--[U])t, L+U^-t]$ the solution
$U_{FP}$ contains a rarefaction wave, however this part of solution
does not contribute to the energy, since for the domain of
rarefaction $\rho=H=0$. The respective solution $(\rho, U)$ to the
sticky particles model is
$$
\rho(x,t)=\rho^-- [\rho]\Theta(x-x_j(t))+e(t)\delta(x-x_j(t)),
\qquad
$$
$$
U(x,t)=U^-- [U]\Theta(x-x_j(t)),
$$
where the position of the singularity  $x_j(t)$ is the following:
\begin{equation}
\label{x(t)[f]}
x_j(t)=\frac{[U\rho]\,
-\,\sqrt{(-[U\rho]^2+[\rho][U^2\rho])}}{[\rho]}\,t, \quad
\mbox{if} \quad [\rho]\ne 0,
\end{equation}
and
\begin{equation}
\label{x(t)[0]}
x_j(t)\,=\,\frac{2U^--[U]}{2}\,t\,=\,\frac{U^-+U^+}{2}\,t,
 \quad \mbox{if} \quad [\rho]\,=\,0.
\end{equation}
The amplitude of the $\delta$-shock reads as
$$
e(t)=[U\rho]\,t\,-\,[\rho]\,x_j(t).
$$
The solution induces the following balance of energy.

For the sake of simplicity we dwell on the latter case. Thus,
$$
W_{int}(t)=(x_j(t) +L-U^-t)H^-+(L+(U^--[U])t-x_j(t))H^+
\qquad\qquad
$$
$$
=W_{int}(0)-\frac{H^++H^-}{2} \,[U]\,t,
\qquad\qquad\qquad\qquad\qquad\qquad
$$
$$
W_{kin}(t)=\frac{1}{2}\Big((x_j(t) +L-U^-t)\rho^-(U^-)^2
\qquad\qquad
$$
$$
\qquad\qquad\qquad\qquad
+(L+(U^--[U])t-x_j(t))\rho^+(U^+)^2\Big)
$$
$$
=W_{kin}(0)-\rho [U]\frac{(U^1)^2+(U^2)^2}{4}\,t,
$$
$$
w_{kin}=\frac{1}{2}[U]\rho \left(\frac{(U^1)+(U^2)}{2}\right)^2\,t,
$$
$$
w_{int}=W_{int}(0)+W_{kin}(0)-W_{int}(t)-W_{kin}(t)-w_{kin}(t)=$$$$
\frac{[U]}{2}\,\left(\frac{\rho[U]^2}{4}+(H^-+H^+)\right)\,t.
$$
We see that $W_{kin}(t)$ and $W_{int}(t)$ decrease with a constant
velocity, and vanish within a finite time, $w_{kin}{t}$ increases
unless $U^-+U^+\ne 0,$ (in the latter case $w_{kin}(t)\equiv 0$),
$w_{int}(t)$ increases {\it in any case}. Since we associate the
internal energy with a temperature, it signifies that the
concentration process {\it always} entails the heating of point of
the mass accumulation and cooling-down of the environment to the
"absolute zero" that relates to the zero internal energy.

\appendix

\section{Some auxiliary facts}
\label{s6}

\subsection{Moving surfaces of discontinuity.}\label{s5.1}
Let us present some results from~\cite[5.2.]{Kan} concerning moving surfaces.
Let $\Gamma_t$ be a smooth moving surface of codimension~1 in the
space $\bR^n$. Such a surface can be represented locally either
in the form $\Gamma_t=\bigl\{x\in \bR^n: S(x,t)=0\bigr\}$, or in
terms of the curvilinear Gaussian coordinates $s=(s_1,\dots,s_{n-1})$
on the surface:
$$
x_j=x_j(s_1,\dots,s_{n-1},t), \qquad s\in \bR^{n-1}.
$$
We also consider the surface $\Gamma=\bigl\{(x,t)\in \bR^{n+1}: S(x,t)=0\bigr\}$
as a submanifold of the space-time $\bR^n\times\bR$.
We shall assume that $\nabla S(x,t)\bigr|_{\Gamma_t}\neq 0$ for all fixed values
of $t$, where $\nabla
=\big(\frac{\partial}{\partial x_1},\dots,\frac{\partial}{\partial x_n}\big)$.
Let $\nu$ be the unit space normal to the surface $\Gamma_t$ pointing in
the positive direction such that
$\frac{\partial S}{\partial x_j}=|\nabla S|\nu_j$, \ $j=1,\dots,n$.

Let $f(x,t)$ be a function defined on the surface $\Gamma_t$
for some time interval, and denote by $\frac{\delta f}{\delta t}$
the derivative with respect to time $t$ as it would be computed by an
observer moving with the surface. This derivative has the following
geometrical interpretation. Let $M_0$ be a point on the surface at
the time $t=t_0$. Construct the normal line to the surface at $M_0$.
At the time $t=t_0+\Delta t$, \, $\Delta t$ is sufficiently small, this normal
meets the surface $\Gamma_{t_0+\Delta t}$ at the point $M=M(t_0+\Delta t)$.
Then the $\delta$-derivative is defined as
\begin{equation}
\label{g-70}
\frac{\delta f(M_0,t_0)}{\delta t}=\lim_{\Delta t \to 0}
\frac{f(M)-f(M_0)}{\Delta t}.
\end{equation}
If $\Delta s$ is the distance between $M_0$ and $M$, then
\begin{equation}
\label{g-71}
G=\lim_{\Delta t \to 0}\frac{\Delta s}{\Delta t}
\end{equation}
is the {\em normal velocity of the moving surface} $\Gamma_t$ and
\begin{equation}
\label{g-72}
\frac{\delta x_j}{\delta t}=\lim_{\Delta t \to 0}\frac{\Delta x_j}{\Delta t}
=\lim_{\Delta t \to 0}\frac{\Delta s}{\Delta t}\frac{\Delta x_j}{\Delta s}
=G\nu_j, \quad j=1,\dots,n.
\end{equation}

Since it is essential that the $\delta$-derivative is
computed on a surface, and $S$ remains constant on this surface, then
$\frac{\delta S}{\delta t}=0$. Thus we have
$$
0=\frac{\delta S}{\delta t}=\frac{\partial S}{\partial t}
+\sum_{j=1}^n\frac{\delta S}{\delta x_j}\frac{\delta x_j}{\delta t}
=\frac{\partial S}{\partial t}+\sum_{j=1}^nG|\nabla S|\nu_j^2,
$$
i.e.,
\begin{equation}
\label{g-73}
S_t=-G|\nabla S|.
\end{equation}
From this formula we can see that $-G=\frac{S_t}{|\nabla S|}$ can be
interpreted as the time component of the normal vector.

The space-time unit normal to the surface $\Gamma$ is given by
${\bf n}=\frac{(\nu,-G)}{\sqrt{1+G^2}}$, where
$\sqrt{1+G^2}=\frac{|\nabla_{(x,t)} S|}{|\nabla S|}$, \
$\nabla_{(x,t)}=\big(\nabla,\frac{\partial}{\partial t}\big)$.

If $f(x,t)$ is a function defined only on $\Gamma$, its first
order $\delta$-derivatives with respect to the time and space variables
are defined by the following formulas~\cite[5.2.(15),(16)]{Kan}:
\begin{equation}
\label{g-74}
\frac{\delta f}{\delta t}\stackrel{def}{=}\frac{\partial \widetilde{f}}{\partial t}
+G\frac{\partial\widetilde{f}}{\partial\nu},
\qquad
\frac{\delta f}{\delta x_j}\stackrel{def}{=}\frac{\partial\widetilde{f}}{\partial x_j}
-\nu_j\frac{\partial\widetilde{f}}{\partial\nu},
\quad j=1,\dots,n,
\end{equation}
where $\widetilde{f}$ is a smooth extension of $f$ to a neighborhood
of $\Gamma$ in $\bR^n\times\bR$, \ $j=1,\dots,n$, and
$\frac{\partial\widetilde{f}}{\partial\nu}=\nu\cdot\nabla\widetilde{f}$
is the normal derivative.
Since $\delta$-derivatives are independent of the way of extension of the
function $f$, we shall drop tilde from the function $f$.
Thus the gradient tangent to the surface $\Gamma_t$ is defined as
\begin{equation}
\label{g-74-3}
\nabla_{\Gamma_t}=\nabla-\nabla_{\nu}
=\Big(\frac{\delta}{\delta x_1},\dots,\frac{\delta}{\delta x_n}\Big),
\end{equation}
where $\nabla_{\nu}=\nu\big(\nu\cdot\nabla\big)$ is the gradient
along the normal direction to $\Gamma_t$.
The {\it mean curvature} of the surface $\Gamma_t$ is defined as
\begin{equation}
\label{g-84.3}
{\cK}\stackrel{def}{=}-\frac{1}{2}\nabla_{\Gamma_t}\cdot\nu
=-\frac{1}{2}\sum_{j=1}^n\frac{\delta \nu_j}{\delta x_j}
=-\frac{1}{2}\nabla\cdot\nu.
\end{equation}

\subsection{Distributions defined on a surface.}\label{s5.2}
The Heaviside function $H(S)$ is introduced by the following definition:
$$
\big\langle H(S), \ \varphi(\cdot,\cdot)\big\rangle=\int_{S \ge 0}\varphi(x,t)\,dx\,dt,
\quad \forall \, \varphi \in {\cD}(\bR^n\times\bR).
$$
According to~\cite[5.3.(1),(2)]{Kan}, we introduce the delta function
$\delta(S)$ on the surface $\Gamma$:
\begin{equation}
\label{g-106}
\big\langle \delta(S), \ \varphi(\cdot,\cdot) \big\rangle
=\int_{-\infty}^{\infty}\int_{\Gamma_t}\varphi(x,t)\,d\Gamma_t\,dt
=\int_{\Gamma}\varphi(x,t)\frac{\,d\Gamma}{\sqrt{1+G^2}},
\end{equation}
for all $\varphi\in {\cD}(\bR^n\times\bR)$, where $d\Gamma_t$ and
$d\Gamma$ are the surface measures on the surfaces $\Gamma_t$ an $\Gamma$,
respectively.
According to~\cite[5.5.Theorem~1.]{Kan}, we have
\begin{equation}
\label{g-84*-h}
\frac{\partial H(S)}{\partial x_j}=\nu_j\delta(S),
\qquad
\frac{\partial H(S)}{\partial t}=-G\delta(S).
\end{equation}

\subsection{An integration-by-parts formula.}\label{s5.3}
We need the following integrating-by-parts formula.

\begin{lemma}
\label{g-lem4}
{\rm(~\cite[Lemma~9.1.]{umn}, cf.~\cite[5.2.(25),(26)]{Kan})}
Suppose that $a(x,t)$ is a smooth function defined
only on the surface $\Gamma=\bigl\{(x,t):S(x,t)=0\bigr\}$ which is
the restriction of some smooth function defined in a neighborhood
of $\Gamma$ in $\bR^n\times\bR$, and $\Gamma_0=\bigl\{x: S(x,0)=0\bigr\}$.
Then the following formula for integration by parts holds:
\begin{equation}
\label{g-84.20}
\int_{\Gamma}a \frac{\delta\varphi}{\delta t}\frac{\,d\Gamma}{\sqrt{1+G^2}}=
-\int_{\Gamma}\frac{\delta^* a}{\delta t}\varphi\frac{\,d\Gamma}{\sqrt{1+G^2}}
-\int_{\Gamma_0}a(x,0)\varphi(x,0)\,d\Gamma_{0},
\end{equation}
for any $\varphi \in {\cD}(\bR^n\times [0, \infty))$,
where $\frac{\delta^*}{\delta t}$ is the adjoint operator defined as
\begin{equation}
\label{g-84.20-1}
\frac{\delta^* a}{\delta t}=\frac{\delta a}{\delta t}-2{\cK}Ga
=\frac{\delta a}{\delta t}+\nabla_{\Gamma_t}\cdot(aG\nu),
\end{equation}
${\cK}$ is the mean curvature {\rm(\ref{g-84.3})} of the surface $\Gamma_t$.
\end{lemma}

\subsection{Transport theorems.}\label{s5.4}
Here we give the following {\it transport theorems\/}.

\begin{theorem}
\label{g-th8}
{\rm(\cite[12.8.(3)]{Kan})}
Let $f(x,t)$ be a sufficiently smooth function defined in a moving solid
$\Omega_t$, and let a moving hypersurface $\partial \Omega_t$ be its boundary.
Let $\nu$ be the outward unit space normal to the surface $\partial \Omega_t$
and $V(x,t)$ be the velocity of the point $x$ in $\Omega_t$. Then the volume
transport theorem holds:
$$
\frac{d}{dt}\int_{\Omega_t}f(x,t)\,dx
=\int_{\Omega_t}\frac{\partial f}{\partial t}\,dx
+\int_{\partial \Omega_t}f V\cdot\nu \,d\Gamma_{t}
\qquad\quad
$$
\begin{equation}
\label{g-84.1*}
\qquad
=\int_{\Omega_t}\Big(\frac{\partial f}{\partial t}+{\rm div}(fV)\Big)\,dx,
\end{equation}
where $d\Gamma_t$ is the surface measure on the moving surface $\partial \Omega_t$.
\end{theorem}

\begin{theorem}
\label{g-th4-tr}
{\rm(~\cite[12.8.(9)]{Kan})}
If $e(x,t)$ is a smooth function defined only on the moving surface
$\Gamma_t=\bigl\{x: S(x,t)=0\bigr\}$ {\rm(}which is the restriction of
some smooth function defined in a neighborhood of $\Gamma_t${\rm)},
then the surface transport theorem holds:
$$
\frac{d}{dt}\int_{\Gamma_t}e(x,t)\,d\Gamma_{t}
\qquad\qquad\qquad\qquad\qquad\qquad\qquad\qquad\qquad\qquad
$$
\begin{equation}
\label{g-84.20-tr}
=\int_{\Gamma_t}\Big(\frac{\delta e}{\delta t}-2{\cK} Ge\Big)\,d\Gamma_{t}
=\int_{\Gamma_t}\Big(\frac{\delta e}{\delta t}
+\nabla_{\Gamma_t}\cdot(eU_{\delta})\Big)\,d\Gamma_{t},
\end{equation}
where $U_{\delta}=\nu G$ is the velocity of $\Gamma_{t}$ given by~{\rm(\ref{g-4.0-10*})}.
\end{theorem}

\bibliographystyle{amsalpha}

\begin{thebibliography}{A}

\bibitem{AKR}{S.Albeverio, A.Korshunova, O.Rozanova,} \textit{Probabilistic model
associated with the pressureless gas dynamics,} submitted, E-print
arXiv:0908.2084.


\bibitem{Al-S}
S. Albeverio, V. M. Shelkovich,
\textit{On the delta-shock front problem},
in the book: ``Analytical Approaches to Multidimensional
Balance Laws'', Ch. 2, (Ed. O.~S.~Rozanova), Nova Science Publishers,
Inc., 2005, pp. 45--88.

\bibitem{B-D-D-R}
F. Berthelin, P. Degond, M. Delitala, M. Rascle,
\textit{A model for the formation and evolution of traffic jams},
Arch. Rat. Mech. Anal., {\bf 187}, Issue 2, (2008), 185--220

\bibitem{B}
F. Bouchut,
\textit{On zero pressure gas dynamics},
Advances in Math. for Appl. Sci., World Scientific, {\bf 22}, (1994),
171--190.

\bibitem{Chen-Liu2}
G.~Q.~Chen, H. Liu,
\textit{Concentration and cavitation in the
vanishing pressure limit of solutions to the Euler equations
for nonisentropic fluids},
Physica D, {\bf 189}, (2004), 141--165.

\bibitem{D-S4}
V. G. Danilov, V. M. Shelkovich,
\textit{Delta-shock wave type solution of hyperbolic systems of conservation
laws},
Quarterly of Applied Mathematics, {\bf 63}, no.~3, (2005), 401--427.

\bibitem{E-R-S}
Weinan E, Yu. Rykov, Ya. G. Sinai,
\textit{Generalized variational principles, global weak solutions and
behavior with random initial data for systems of conservation laws arising
in adhesion particle dynamics},
Comm. Math. Phys., {\bf 177}, (1996), 349--380.

\bibitem{Oc-Oc}
C.~M.~Edwards, S.~D.~Howinson, H.~Ockendon and J.~R.~Ockendon,
\textit{Non-classical shallow water flows},
Journal of Applied Mathematics, {\bf 73}, (2008), 137--157.

\bibitem{Evans}
L.~C.~Evans,
{\em Partial Differential Equations},
Amer. Math. Soc. Providence, Road Island, 1998.

\bibitem{F-Meerson-2}
I. Fouxon, B. Meerson, M. Assaf, and E. Livne,
\textit{Formation of density singularities in ideal hydrodynamics
of freely cooling inelastic gases: A family of exact solutions},
Phys. Fluids, {\bf 19}, 093303 (2007), (17 pages).

\bibitem{F-Meerson-3}
I. Fouxon, B. Meerson, M. Assaf, and E. Livne,
\textit{Formation of density singularities in hydrodynamics of inelastic gases},
Phys. Review, E {\bf 75}, 050301(R) (2007), (4 pages).

\bibitem{Kan}
Ram P. Kanwal,
\textit{Generalized Functions: Theory and technique},
Birkh\"{a}user Boston--Basel--Berlin, 1998.

\bibitem{Korchinski}
D.~J.~Korchinski,
{\em Solution of a Riemann problem for $2\times2$ systems of conservation
laws possesing no classical weak solution},
Ph.D. Thesis, Adelphi Univ., Garden City, N. Y., 1977.

\bibitem{Kraiko}
A. N. Kraiko,
\textit{Discontinuity surfaces in medium without self-pressure},
Prikladnaia Matematika i Mekhanika, {\bf 43}, (1979), 539--449.
(In Russian)

\bibitem{Kraiko-1}
A. N. Kraiko,
\textit{On two-phase flows model of gas and dispersed in it particles},
Prikladnaia Matematika i Mekhanika, {\bf 46}, issue~1, (1982), 96--106.
(In Russian)

\bibitem{Kraiko-2}
A. N. Kraiko, S. M. Sulaimanova,
\textit{Two-phase flows of a gas-particle mixture near impermeable surfaces
with the formation of ``sheets'' and ``filaments},
Prikladnaia Matematika i Mekhanika, {\bf 47}, issue~4, (1983), 619--630.
(In Russian)

\bibitem{Li-Zh1}
J. Li, Tong Zhang,
\textit{On the initial-value problem for zero-pressure gas dynamics},
Hyperbolic problems: Theory, Numerics, Applications.
Seventh International Conference in Z\"{u}rich, February 1998,
Birkh\"{a}user Verlag, Basel, Boston, Berlin, 1999, 629--640.

\bibitem{Li-Y}
J. Li, Hanchun Yang,
\textit{Delta-shocks as limit of vanishing viscosity for multidimensional
zero-pressure gas dynamics},
Quart. Appl. Math., {\bf LIX}, N 2, (2001), 315--342.

\bibitem{Ned-1}
M.~Nedeljkov,
{\em Shadow Waves: Entropies and Interactions for Delta and Singular Shocks},
Archive for Rational Mechanics and Analysis, (2010).

\bibitem{Ned-Oberg}
M.~Nedeljkov, M.~Oberguggenberger,
{\em Interactions of delta shock waves in a strictly hyperbolic system of conservation laws},
Journal of Mathematical Analysis and Applications,
{\bf 344}, Issue 2, (2006), 1143--1157.

\bibitem{Osiptsov-1}
A. N. Osiptsov,
\textit{Investigation of regions of unbounded growth of the
particle concentration in dispersi flows},
Fluid Dynamics, {\bf 19}, (1984), no.~3, 378--385.

\bibitem{Osiptsov-2}
A. N. Osiptsov,
\textit{Modified Lagrangian method for calculating the particle concentration
in dusty-gas flows with intersecting particle trajectories}, Proc. 3d
Intern. Conf. Multiphase Flows, Lyon, France, CD-ROM "ICMF'98", 1998,
paper 236, 8 p.

\bibitem{Osiptsov-3}
A. N. Osiptsov,
\textit{Lagrangian modeling of dust admixture in gas flows},
Astrophys. Space Sci., {\bf  274}, (2000), 377--386.

\bibitem{Pan-S1}
E. Yu. Panov, V. M. Shelkovich,
\textit{$\delta'$-Shock waves as a new type of solutions to
systems of conservation laws},
Journal of Differential Equations, {\bf 228 }, (2006), 49--86.

\bibitem{San-Z}
S. F. Shandarin and Ya. B. Zeldovich,
\textit{The large-scale structure of the universe: turbulence,
intermittence, strucrures in self-gravitating medium},
Rev. Mod. Phys., {\bf 61}, (1989), 185--220.

\bibitem{S-2008}
V.~M.~Shelkovich,
\textit{Transport of mass, momentum and energy in zero-pressure gas dynamics}
in: Proceedings of Symposia in Applied Mathematics 2009; Volume: 67.
Hyperbolic Problems: Theory, Numerics and Applications
Edited by: E. Tadmor, Jian-Guo Liu, and A.E. Tzavaras, AMS, 2009.
929--938.

\bibitem{S5}
V. M. Shelkovich,
\textit{The Riemann problem admitting $\delta$-, $\delta'$-shocks,
and vacuum states {\rm(}the vanishing viscosity approach{\rm)}},
Journal of Differential Equations, {\bf 231}, (2006), 459--500.

\bibitem{S8}
V. M. Shelkovich,
\textit{The Rankine--Hugoniot conditions and balance laws for
$\delta$-shocks},
Fundamentalnaya i Prikladnaya Matematika, v.~12, no.~6, (2006), 213--229.
(In Russian).
English transl. in: Journal of Mathematical Sciences,
Springer US, v.~151, (2008), no.~1, 2781--2792.

\bibitem{umn}
V. M. Shelkovich,
\textit{$\delta$- and $\delta'$-shock types of singular solutions
to systems of conservation laws and the transport and concentration processes},
Uspekhi Mat. Nauk, {\bf 63}:3(381), (2008), 73--146.
English transl. in Russian Math. Surveys, {\bf 63}:3, (2008), 473--546.

\bibitem{S-Zh}
Wancheg Shen, Tong Zhang,
\textit{The Riemann problem for the transportaion equations in gas dynamics},
Memoirs of the Amer. Math. Soc., {\bf 137}, no.~654, (1999), 1--77.

\bibitem{Y-1}
Hanchun Yang,
\textit{Generalized plane delta-shock waves for $n$-dimensional zero-pressure
gas dynamics},
Journal of Mathematical Analysisi and Applications, {\bf 260}, (2001), 18--35.

\bibitem{Z}
Ya. B. Zeldovich,
\textit{Gravitationnal instability: An approximate theory for large
density perturbations},
Astron. Astrophys., {\bf 5}, (1970), 84--89.

\bibitem{Zeldovich-Mishkis}
Y.~B.~Zeldovich, A.~D.~Myshkis,
{\em Elements of mathematical physics.
Medium consisting of non-interacting particles.},
M.: Nauka, 1973. (In Russian)


\end{thebibliography}

\end{document}